%% file: main.tex
\documentclass[conference,a4paper]{IEEEtran}
\usepackage[utf8]{inputenc} 
\usepackage[T1]{fontenc}
\usepackage{url}              % provides \url{...}
\usepackage{algorithm}
\usepackage{hyperref}
\usepackage{multirow}
\usepackage{bbding}
\usepackage{color}
\usepackage{amsmath,amssymb,amsfonts}
\usepackage{amsthm}
\usepackage{graphicx}
\usepackage{setspace}
\usepackage{multicol}
\usepackage{multirow}
\usepackage{cite}
\usepackage{float}
\usepackage{color}
\usepackage{alltt, listings}
\usepackage{subfigure}
\usepackage{algpseudocode}
\usepackage{wrapfig}
\usepackage{url}
\usepackage{bm}
\usepackage{amsmath}
\usepackage{mleftright}
\usepackage{authblk}
\mleftright

\newtheorem{assumption}{Assumption}
\newtheorem{theorem}{Theorem}
\newtheorem{lemma}{Lemma}

\theoremstyle{definition}
\newtheorem{definition}{Definition}

\begin{document}

\title{Towards Quantum-Safe Federated Learning via Homomorphic Encryption: Learning with Gradients} 

% \author{
%   \IEEEauthorblockN{Guangfeng Yan, Shanxiang Lyu, Linqi Song, Zhiyong Zheng}}

\author[a]{Guangfeng Yan}
\author[b]{Shanxiang Lyu}
\author[a]{Hanxu Hou}
\author[d]{Zhiyong Zheng}
\author[a]{Linqi Song}
\affil[a]{Department of Computer Science, City University of Hong Kong, Hong Kong SAR, China}
\affil[b]{College of Cyber Security, Jinan University, Guangzhou, China}
% \affil[c]{School of Engineering and Intelligentization, Dongguan University of Technology, Dongguan, China}
\affil[c]{School of Mathematics, Renmin University of China, Beijing, China}

% \footnote{*Guangfeng Yan and Shanxiang Lyu contributed to the work equally and should be regarded as co-first authors.}
\maketitle

\begin{abstract}
This paper introduces a privacy-preserving distributed learning framework via private-key homomorphic encryption. Thanks to the randomness of the quantization of gradients, our learning with error (LWE) based encryption can eliminate the error terms, thus avoiding the issue of error expansion in conventional LWE-based homomorphic encryption.
The proposed system allows a large number of learning participants to engage in neural network-based deep learning collaboratively over an honest-but-curious server, while ensuring the cryptographic security of participants' uploaded gradients. 
\end{abstract}


\input{1.Introduction}
\input{2.Preliminaries}

\input{3.PEFL}
\input{4.Experiments}
\input{5.Conclusion}
\clearpage
\bibliographystyle{IEEEtran}
\bibliography{mybib}
% \clearpage
% \input{Appendix}

\end{document}

%% file: 1.Introduction.tex
\section{Introduction}
\label{introduction}
Federated learning (FL) comprises a family of decentralized training algorithms for machine learning~\cite{dean2012large,bekkerman2011scaling,mcmahan2016federated}, enabling individuals to collaboratively train a model without centralizing the training data. This approach alleviates the computational burden on data centers by distributing training computation to the edge. However, it is crucial to note that while federated learning offers a decentralized framework, it may not inherently safeguard the privacy of clients. The updates received by the central server have the potential to inadvertently reveal information about the client's training data~\cite{DBLP:journals/tifs/PhongAHWM18,DBLP:journals/tdsc/LiuWLPW23}.

Popular strategies to protect the privacy of clients for federated learning include differential privacy (DP) based and homomorphic encryption (HE) based methods.
The idea of DP is to add noises to the gradients to protect the secrecy of gradients~\cite{abadi2016deep}. Existing works on DP based learning algorithms include local DP (LDP)~\cite{erlingsson2014rappor}, DP with selective parameter updates\cite{shokri2015privacy}, DP based on lattices~\cite{stevens2022efficient} etc.
Although DP can be adopted in a straightforward manner, it has the downside of weaker privacy guarantee and potential accuracy loss.
 
HE is a cryptographic technique that enables computations to be performed on encrypted data without the need to decrypt it first. In the context of federated learning, homomorphic encryption plays a crucial role in ensuring the privacy of individual participants' data. Since the aggregation of gradients in FL only involves addition, many recent works~\cite{zhang2020batchcrypt,DBLP:conf/icc/ZhangFWZC20} have proposed to employ additively homomorphic encryption based on Paillier~\cite{paillier1999public}.
However, Paillier's security is broken as soon as one can efficiently factor large integers using Shor's quantum algorithm~\cite{shor1999polynomial}. 

Lattice-based cryptography is considered quantum-resistant~\cite{peikert2014lattice,wang2022quantum,saliba2021reconciliation}. Certain lattice-based problems, such as the Learning With Errors (LWE) problem~\cite{regev2009lattices}, are believed to be hard for quantum computers to solve efficiently. Compared to lattice-based fully homomorphic encryption \cite{gentry2009fully,zheng2022fully}, lattice-based additively homomorphic encryption~\cite{DBLP:journals/tifs/PhongAHWM18} is considered a promising approach for FL, as only addition is needed in the aggregation of gradients. However, the scheme in~\cite{DBLP:journals/tifs/PhongAHWM18} assumes a shared negotiation of common public-private key pairs among all clients, raising concerns about potential data leakage between clients. Even when limited to linear functions, exisitng HE schemes~\cite{DBLP:journals/tifs/PhongAHWM18} do  not allow to perform an arbitrary number of additions, because, each time two ciphertexts are added, the error gets bigger.
 
 This study aims to devise an HE scheme where each client possesses an individual private key. The unique contributions of this work can be succinctly summarized as follows:
\begin{itemize}
 \item \textbf{Post-Quantum Security:} Our system guarantees the non-disclosure of participants' information to the honest-but-curious parameter (cloud) server. The inherent security, stemming from lattice-based computational problems, remains robust even in the quantum era. Demonstrating comparable accuracy to a federated learning system trained on the joint dataset of all participants, our system attains identical precision in its predictive outcomes.

 \item \textbf{Aggregation with High Accuracy:} By introducing the concept of "learning with gradients," our approach eliminates the error term present in conventional Learning With Errors (LWE)-based encryption. Furthermore, the incorporation of randomized signed quantization significantly reduces the probability of overflow arising from the summation of gradients.
    
\item \textbf{Small Communication Cost:} The core of our encryption mechanism involves utilizing $\bm{As}$ in LWE as ``masking''. This design results in a small communication factor (i.e. ciphertext expansion ratio), ensuring an efficient and expedient data transfer process. When the number of quantization bits is $b=6,8,10$, the increased communication factors $\tau = 3.5, 2.876, 2.5$ respectively.
\end{itemize}

%% file: 2.Preliminaries.tex
\section{Preliminaries}

We consider a distributed learning problem, where $N$ clients collaboratively participate in training a shared model via a central server.
The local dataset located at client $i$ is denoted as $\mathcal{D}^{(i)}$, and the union of all local datasets $\mathcal{D} = \{\mathcal{D}^{(i)}, i =1,...,N\}$.The objective is to minimize the empirical risk over the data held by all clients, i.e., solve the optimization problem 
\begin{equation}
\begin{array}{llll}	\min_{\bm{\theta} \in \mathbb{R}^d}F(\bm{\theta}) = \cfrac{1}{N} \sum_{i=1}^N  \mathbb{E}_{\xi^{(i)}\sim \mathcal{D}^{(i)}}[l(\bm{\theta};\xi^{(i)})],
	\end{array}\label{optim_problem}
\end{equation}
where $l(\bm{\theta};\xi^{(i)})$ is the local loss function of the model $\bm{\theta}$ towards one data sample $\xi^{(i)}$.

A standard approach to solve this problem is DSGD~\cite{{yan2022ac},{bekkerman2011scaling}}, where each client $i$ first downloads the global model $\bm{\theta}_t$ from server at iteration $t$, then randomly selects a batch of samples $B^{(i)}_t \subseteq D^{(i)}$ with size $B$ to compute its local stochastic gradient with model parameter $\bm{\theta}_t$: $\bm{g}^{(i)}_t = \frac{1}{B} \sum_{\xi^{(i)} \in B^{(i)}_t} \nabla \ell(\bm{\theta}_t;\xi^{(i)})$. Then the server aggregates these gradients and sends the aggregated gradient $\bm{g}_{total}$ back to all clients:$\bm{g}_{total} = \sum_{i=1}^N \bm{g}^{(i)}_t$. Finally, each clients update their local model:
\begin{equation}\label{eq:update_model}
    \bm{\theta}_{t+1} = \bm{\theta}_t - \cfrac{\eta}{N} \bm{g}_{total} 
\end{equation}
where $\eta$ is the learning rate. We make the following two common assumptions on such the raw gradients $\nabla l(\bm{\theta}_t;\xi^{(i)})$ and the objective function $F(\bm{\theta})$~\cite{{bottou2018optimization},{data2023byzantine}}: 
\begin{assumption}[Bounded Variance]
	For parameter $\bm{\theta}_t$, the stochastic gradient $\nabla l(\bm{\theta}_t;\xi^{(i)})$ sampled from any local dataset have uniformly bounded variance for all clients:
	\begin{align}
	\mathbb{E}_{\xi^{(i)}\sim \mathcal{D}^{(i)}}\left[\|\nabla \ell(\bm{\theta}_t;\xi^{(i)})-\nabla F(\bm{\theta}_t)\|^2\right] \le \sigma^2.
	\end{align}
	\label{ass:stochastic_gradient} 
\end{assumption}
\vspace{-5mm}
\begin{assumption}[Smoothness]
	The objective function $F(\bm{\theta})$ is $\nu$-smooth: $\forall \bm{\theta},\bm{\theta}' \in \mathbb{R}^d$, $\|\nabla F(\bm{\theta})-\nabla F(\bm{\theta}')\| \leq \nu\|\bm{\theta}-\bm{\theta}'\|$.
	\label{ass:smoothnesee} 
\end{assumption}
Assumption~\ref{ass:smoothnesee} further implies that $\forall \bm{\theta},\bm{\theta}' \in \mathbb{R}^d$, we have 
\begin{equation}
F(\bm{\theta}') \leq F(\bm{\theta}) + \nabla F(\bm{\theta})^\mathrm{T} (\bm{\theta}'-\bm{\theta}) + \frac{\nu}{2} \|\bm{\theta}'-\bm{\theta}\|^2.
\label{eq:smooth_1}
\end{equation}

%% file: 3.PEFL.tex
\section{Quantum-Safe Federated Learning}

\subsection{The Encryption Scheme}
The security of our encryption scheme relies on the hardness of solving LWE~\cite{peikert2014lattice}.
\begin{definition}[The LWE problem]
Let	$\bm A \leftarrow \mathbb{Z}^{m\times n}_q$, $\bm s \leftarrow\mathbb{Z}^{n}_q$  and $\bm e$ from some error distribution $\chi_{\sigma}^{m}$. \begin{itemize}
    \item The search-LWE problem is,
given $(\bm A, \bm{As+e})$,  find $\bm s$.
\item The decision-LWE problem is, given $(\bm A, \bm{As+e})$, distinguish it from $(\bm A, \bm{u})$, $\bm u \leftarrow\mathbb{Z}^{m}_q$.
\end{itemize} 
\end{definition}
In general, the entries of $\bm e$
are i.i.d. from a Gaussian-like distribution with standard deviation $\sigma\geq \sqrt{n}$. However, the hardness of LWE retains also for other types of small errors (e.g., uniform binary errors), provided that the number of samples is linear over $n$ \cite{micciancio2013hardness}.

%To avoid trivial solution for search-LWE, one should have $m\geq n$.

Our idea is inspired by the learning with rounding (LWR)~\cite{alwen2013learning} problem, where the error term has been eliminated thanks to the randomness induced by modulus reduction. In a similar vein, if the quantization of gradients can induce random errors, the error term in LWE can also be cancelled. Let $\bm A$ be public and $\bm s$ be the secret key. We propose to encrypt the gradients by using
\begin{align}\label{eq:encryption}
 \bm{ct} \triangleq \mathrm{Enc}_{\bm s}(\bm g) = \bm {A s} + \gamma \mathrm{Q}_b({\bm g}) \in \mathbb{Z}^{m}_q
\end{align}
where $\mathrm{Q}_b(\cdot)$ denotes a quantizer using $b$ bits, and $\gamma$ denotes a proper scaling that ensures LWE is hard enough \footnote{In a conservative manner, $\gamma \mathrm{Q}_b({\bm g})$ lies in  $\mathbb{Z}^{m}_q$.}. 

Regarding decryption, we have
\begin{align}\label{eq:decryption}
	\bm g = \mathrm{Dec}_{\bm s}(\bm{ct}) = \gamma^{-1}( \mathrm{Enc}_{\bm s}(\bm g) - {\bm A \bm s}). 
\end{align}

In the distributed SGD model, we can assume that each client keeps a secret key ${\bm s}^{(i)}$. Via Shamir's secret sharing protocol~\cite{shamir1979share}, they can agree on the vector sum $\bm{s}_{sum} = \sum_{i=1}^{N}\bm{s}^{(i)}$.  The server aggregates all encrypted gradient as

\begin{align}\label{eq:aggregated_gra}
    \bm {ct}_{total}= \sum_{i=1}^N \mathrm{Enc}_{\bm s^{(i)}}(\bm g^{(i)}) 
	 =  {\bm A} \sum_{i=1}^N  {\bm s^{(i)}}  +  \gamma \sum_{i=1}^N \mathrm{Q}_b(\bm g^{(i)}),
\end{align}
and broadcast it to all clients. Each client can decrypted the aggregated gradients:
\begin{align}\label{eq:decrypted_gra}
    \bm {g}_{total} = \gamma^{-1}(\bm {ct}_{total} - {\bm A} \bm{s}_{sum})
\end{align}
The $\bm {g}_{total}$ is the same as that of the non-encryption based aggregation. Thus
our system features the strengths of cryptographic security with the precision of deep learning accuracy, offering the best of both worlds.

The whole procedure is summarized in Algorithm~\ref{alg:FLAG}, referred to as \textbf{F}ederated \textbf{L}earning via le\textbf{A}rning with \textbf{G}radients (FLAG), that incorporates Private-key Encryption into a distributed SGD framework.

\begin{algorithm}[th!] 
	\caption{Federated Learning via leArning with Gradients (FLAG)}
	\begin{algorithmic}[1]
	\State \textbf{Input:} Learning rate $\eta$, initial point $\bm{\theta}_0 \in \mathbb{R}^d$, communication round $T$, quantization bit $b$;
		\For {each communication rounds $t = 0, 1, ..., T-1$:}
		\State \textbf{On each client {$i=1, ..., N$}:}
        \State Receive the aggregated ciphertext and use the Eq.~\eqref{eq:decrypted_gra} to decrypte $\bm {g}_{total}$;
        \State Update the local model using~Eq.\eqref{eq:update_model};
		\State Compute the local gradient $\bm{g}^{(i)}_t$ using SGD;
        \State Clip and quantize $\bm{g}^{(i)}_t$ to $\mathrm{Q}_b(\bm g^{(i)})$ using Eq.~\eqref{eq:clipoperation} and \eqref{eq:quantize};
		\State Encrypte $\bm g^{(i)}$ as $\bm {ct}^{(i)}$ using Eq.~\eqref{eq:encryption};
		\State Send $\bm {ct}^{(i)}$ to the server;
 	\State \textbf{On the server:}
        \State Aggregate all encrypted gradients using Eq.~\eqref{eq:aggregated_gra};
		\State Broadcast aggregated ciphertext to all clients;
	\EndFor
\end{algorithmic} 
	\label{alg:FLAG}
\end{algorithm}
 
\subsection{On randomized quantization}
In the preprocessing stage, We clip the raw gradient $\bm{g}^{(i)}$ into $l_{\infty}$ norm with threshold $C$: 
\begin{equation}
	\bm{\hat g}^{(i)} = \frac{\bm{g}^{(i)}}{\max{\{1, \|\bm{g}^{(i)}\|_{\infty}/C\}}}. \label{eq:clipoperation} 
\end{equation}

%In this work, we form $Q_b[\cdot]$ using a two-stage quantizer. 

We then quantize the clipped gradient into $b$ bits using the half-dithered quantizer~\cite{gray1993dithered,abdi2019nested}. In particular, the element-wise quantization function $Q_b$ is defined as
%the $j$-th element of $\bm{\hat g}^{(i)}$ is quantized as:
\begin{align}\label{eq:quantize}
   Q_b(\bm{\Tilde g}^{(i)}_j) = \text{round} (\cfrac{\bm{\hat g}_j^{(i)}}{\Delta} + u \bm I) \cdot \Delta,
\end{align}
where $\Delta = \frac{2C}{2^b}$ is the quantization step, $u \sim U(-\frac{\Delta}{2},\frac{\Delta}{2})$ is a uniform dither signal. Note that the scaling factor in encyprion is set as $\gamma = \frac{q}{2C}$. The main feature of the half-dithered quantizer is that $u$ is not needed at the receiver's side.
%and $\text{round}$ is the round operation. 

With reference to~\cite{gray1993dithered}, the  quantization noise $\bm{\varepsilon}= Q_b(\bm{\Tilde g}^{(i)}_j) - \bm{\Tilde g}^{(i)}_j$ equals to the sum of an uniform random variable $\bm{e}'\triangleq U(-\frac{\Delta}{2},\frac{\Delta}{2})$ and a determined dither signal $\bm{e}''\triangleq u\bm I$. It can be verified that $\mathbb{E}(\bm{\varepsilon}) = \bm{0}$, $\mathbb{E}(\bm{\varepsilon}_i)^2 \le \cfrac{C^2}{2^{2b-2}}.$

%and the second moment of the error term is
%\begin{equation}
%	\mathbb{E}(\bm{\varepsilon}_i')^2 \le \cfrac{C^2}{2^{2b-2}}. 
%\end{equation}
%It is common in FL literature to assume that the Gradients in FL admit Gaussian distributions\cite{zhang2020batchcrypt}. In this case, we have $\mathbb{E}(\bm{e}_i')^2 = **.$

\subsection{CPA Security}\label{subsection:CPA_Security}
Following~\cite{katz2007introduction}, we introduce the Chosen-Plaintext Attack (CPA) indistinguishability experiment as follows.
\begin{enumerate}
    \item The adversary $\mathcal{A}$ is given oracle access to $\mathrm{Enc}_{\bm s}(\cdot)$ and two random gradients $\bm{g}_0,\bm{g}_1$.
    \item  A random bit $b\leftarrow \lbrace 0,1\rbrace$  is chosen. Then a ciphertext $ \bm{ct} \leftarrow \mathrm{Enc}_{\bm s}(\bm{g}_b)$ is computed and given to $\mathcal{A}$. We call $\bm{ct}$ the challenge ciphertext.
    \item The adversary $\mathcal{A}$ continues to have oracle access to $\mathrm{Enc}_{\bm s}(\cdot)$, and outputs a bit $b'$.
   % \item The output of the experiment is defined to be 1 if $b'=b$, and $0$ otherwise.
\end{enumerate}
 
A secret key encryption scheme is considered CPA secure if, for every efficient adversary $\mathcal{A}$, the following advantage is negligible:
\[
\text{Adv}_{\text{CPA}}(\mathcal{A}) = \left| \Pr[b'= b] - \Pr[b' \neq b] \right|.
\]

Thanks to the half-dither quantizer, we have
\begin{equation}
    \mathrm{Enc}_{\bm s}(\bm{g}_b)=\bm{A} \bm{s} + \gamma \bm{g}_b + \gamma \bm{e}'+ \gamma \bm{e}''
\end{equation}
where $\gamma \bm{e}' \sim U(-\frac{q}{2^{b+1}},\frac{q}{2^{b+1}})^m $ and $\gamma \bm{e}''$ is a determined dither signal. It has been shown in~\cite[Corollary 1]{dottling2013lossy} that LWE with uniform errors is not easier than the standard LWE assumption (with $m\geq 3n$ and $\frac{q}{2^{b+1}}\geq 2n^{0.5+\delta}m$, $\delta\rightarrow 0$). Thus it follows from the indistinguishability of $\bm{A} \bm{s} +  \gamma \bm{e}'$ to a uniform distribution that the CPA security holds. As parameters satisfying the CPA security proof are too strong, in Table~\ref{table:efficiency parameters}, we provide bit-level security based estimates against both classical and quantum attacks.

\subsection{Probability of Overflow}
We further evaluate the probability of overflows in the proposed HE.
We firstly introduce following Lemma to character the tail bound of the mixed Gaussian distribution:
\begin{lemma}
    Let $\bm{x}_{total} = \sum_{i=1}^Nw_i\bm{x}_i$, where $\bm{x}_i \sim N(\bm{0}, \sigma_x^2\bm{I})$ are i.i.d. Gaussian random vectors with zero mean and variance $\sigma_x^2$, and $w_i$ be weights. The tail of $\bm{x}_{total}$ can be approximated as:
    \begin{align}
        Pr(\|\bm{x}_{total}\|_{\infty}>t) = 1 - \Big[1-2erfc\Big(\frac{t}{\sqrt{2\sum_{i=1}^Nw_i^2\sigma_x^2}}\Big)\Big]^N
    \end{align}
\end{lemma}
\begin{proof}
    Given $\bm{x}_i \sim N(\bm{0}, \sigma_x^2\bm{I})$ are i.i.d. Gaussian random vectors with zero mean and variance $\sigma_x^2$, hence $\bm{x}_{total}$ is also Gaussian with mean $\bm 0$ and variance $\sum_{i=1}^Nw_i^2\sigma_x^2\bm{I}$.
\end{proof}
It is common in FL literature to assume that the Gradients in FL admit Gaussian distributions~\cite{zhang2020batchcrypt}. To evaluate the probability of overflows in the proposed HE, we assume that the local gradient $\bm{g}^{(i)} \sim N(\bm{0}, \sigma_g^2\bm{I})$ are i.i.d. Gaussian random vectors with zero mean and variance $\sigma_g^2$. Hence, we have:
\begin{align}
    P_o & \triangleq Pr[\| \gamma \sum_{i=1}^N  Q(\bm g^{(i)})\|_{\infty}>\frac{q}{2}] \nonumber\\
    &=  1 - \Big[1-2erfc\Big(\frac{C}{\sqrt{2N\sigma_g^2}}\Big)\Big]^N
\end{align}
Let $P_o\le \delta$, we have
\begin{align}
    C \ge \sqrt{2N\sigma_g^2}erfc^{-1}\Big(\frac{1}{2}-\frac{1}{2}(1-\delta)^{1/N}\Big)
\end{align}
To prevent overflow while selecting the quantization range, we need to ensure that the value of $C$ is large enough. C represents the size of the selected quantization range. The more clients number $N$ or the smaller $\delta$, the greater the value of $C$ we need to prevent overflow.

\begin{table*}\label{table:efficiency parameters}
\centering
	\caption{The hardness of efficiency parameters. (The `primal' represents the primal attack, the `dual' represents the dual attack, the `classical' represents the known classical complexity, the `quantum' represents the known quantum complexity, the `plausible' represents the best possible algorithm complexity, the `M' represents the number of samples required, and the `K' denotes the block size of BKZ algorithm.)}
\begin{tabular}{|l|l|l|l|l|l|}
\hline
                  parameters ($n$, $q$, $b$) &  attack &  $(M,K)$&  classical&  quantum& plausible \\ \hline
\multirow{2}{*}{($256$, $65536$, $8$)} &primal &$(764,970)$ &294 &267 &211  \\ \cline{2-6} 
                 &dual &$(766,956)$ &290 &263 &208  \\ \hline
\multirow{2}{*}{($256$, $65536$, $9$)} &primal &$(749,716)$ &219 &199 &158  \\ \cline{2-6} 
                 &dual &$(751,707)$ &216 &197 &156  \\ \hline
\multirow{2}{*}{($256$, $65536$, $10$)} &primal &$(654,547)$ &169 &154 &123  \\ \cline{2-6} 
                 &dual &$(667,540)$ &167 &152 &121  \\ \hline
\end{tabular}
\end{table*}

\subsection{Communication Cost}
Noted that in this work, the same $\bm A$ can be reused with many different $\bm s_t$, making the amortized cost of $\bm A$ arbitrarily small~\cite{micciancio2023error}. Hence, pseudo-random generators are implemented on both the server and clients to generate $\bm A$ with a short seed synchronously to improve communication efficiency further.
\begin{lemma}[Increased Communication Factor]
    The communication between the server and clients of FLAG is
\begin{align}
    \tau = 1 + \cfrac{2+1.5\log_2(m)-0.5\log_2(3)}{b}
\end{align}
time of the communication of the corresponding distributed quantized SGD, $b$ is the number of bits of quantization.
\end{lemma}
We can observe that the increased communication factor $\tau$ increases as the value of $m$ increases. In practical, $n$ is usually taken as 128, therefore, we take $m=3n=384$. This results in $\tau = 1 + \frac{12.5 + \log_2(3)}{b}$. This means that the larger the quantization bits, the smaller the increased communication factor. The additional communication overhead is due to the need to transmit additional information to enable homomorphic encryption. By increasing the number of bits, the additional cost is diluted.

\textit{proof outline}: In distributed SGD (Section II-B), each client sends gradients directly to the parameter server at each iteration, so that the communication cost for one iteration in bits is:
\begin{align}
    PlainBits = d\cdot b
\end{align}
In FLAG, we compute the ciphertext length that each client send to the cloud parameter server at each iteration. Hence, its length in bits is $CtBits = d\log_2q$. From Subsection~\ref{subsection:CPA_Security}, we have $q \geq 3^{-0.5}2^{b+2}m^{1.5}$. Hence, the communication cost of FLAG for one iteration in bits is:
\begin{align}
    CtBits = d [b+2+1.5\log_2(m)-0.5\log_2(3)]
\end{align}

Therefore, the increased factor is:
\begin{align}
    \tau \triangleq \cfrac{CtBits}{PlainBits} = \cfrac{b+2+1.5\log_2(m)-0.5\log_2(3)}{b}
\end{align}

Vanilla FL

Then we characterize the convergence performance in the following Theorem.
\begin{theorem}
    For an $N$-client distributed learning problem, the convergence error of FLAG for the smooth objective is upper bounded by
	\begin{align} \label{eq:cov_of_PE-SGD}
		&\frac{1}{T}\sum_{t=0}^{T-1}\mathbb{E}[\|\nabla F(\bm{\theta}_t)\|^2] \le \underbrace{\frac{2[F(\bm{\theta}_0)-F(\bm{\theta}^*)]}{T\eta}+ \frac{\sigma^2}{NB}}_{\triangleq\mathcal{E}_{Vanilla FL}} + \underbrace{\cfrac{dC^2}{N2^{2b}}}_{\triangleq\mathcal{E}_{Q}}
	\end{align}
 And the communication between the server and
clients of FLAG is
\begin{align}
    \tau = 1 + \cfrac{2+1.5\log_2(m)-0.5\log_2(3)}{b}
\end{align}
time of the communication of the corresponding distributed SGD.
\end{theorem}
The first item in Eq.~\ref{eq:cov_of_PE-SGD}, denoted as $\mathcal{E}_{Vanilla FL}$, refers to the convergence error bound of vanilla federated learning without quantization and encryption operations. The second item is the quantization error, indicating the trade-off between communication budget and the accuracy of our algorithm. The quantization error is inversely proportional to $1/2^{2b}$, meaning that less communication budget can lessen the model's accuracy.
\begin{proof} We firstly introduce following Lemma:
\vspace{-2mm}
\begin{lemma}[~\cite{abdi2019nested}]
    The quantization error of the half-dithered quantizer is:
    \begin{align} \label{eq:dithered_error}
    \mathbb{E}[Q_b[\bm{\Tilde g}_t^{(i)}]] = \bm{\Tilde g}_t^{(i)}\\
        \mathbb{E}\|Q_b[\bm{\Tilde g}_t^{(i)}] - \bm{\Tilde g}_t^{(i)}\|^2 \le \cfrac{dC^2}{2^{2b}}
    \end{align}
\end{lemma}
Due to the fact that private-key encryption scheme does not introduce additional errors, Eq.~\eqref{eq:decrypted_gra} acctually is actually equivalent to $\bm {g}_{total} = \sum_{i=1}^N Q_b[\bm{\Tilde g}_t^{(i)}]$. Combining Assumption~\ref{ass:stochastic_gradient} and Eq.~\eqref{eq:dithered_error}, the properties of aggregated gradient $\bm {g}_{total}$ satisfy:
\begin{align}
	&\mathbb{E}[\bm {g}_{total}] = N\nabla F(\bm{\theta}_t),\label{eq:unbiassness}\\
	&\mathbb{E}\left[||\bm {g}_{total}||^2\right] \le N^2\|\nabla F(\bm{\theta}_t)\|^2 
	+ \frac{N\sigma^2}{B} + \cfrac{dNC^2}{2^{2b}}
	\label{eq:qsg}
\end{align}

Firstly, we consider function $F$ is $\nu\text{-smooth}$, and use Eq.~\eqref{eq:smooth_1}:
\begin{align}
	F(\bm{\theta}_{t+1}) &\le F(\bm{\theta}_t) + \nabla F(\bm{\theta}_t)^\top (\bm{\theta}_{t+1}-\bm{\theta}_t) + \frac{\nu}{2} \|\bm{\theta}_{t+1}-\bm{\theta}_t\|^2 \nonumber\\
    &= F(\bm{\theta}_t) + \nabla F(\bm{\theta}_t)^\top (-\cfrac{\eta}{N} \bm {g}_{total}) + \frac{\nu}{2} \|\cfrac{\eta}{N} \bm {g}_{total}\|^2
\end{align}

Taking total expectations and using Eq.~\eqref{eq:unbiassness} and \eqref{eq:qsg}:
\begin{equation}\nonumber
	\begin{split}
		\mathbb{E}[F(\bm{\theta}_{t+1})] &\le F(\bm{\theta}_t) + (-\eta + \frac{\nu \eta^2}{2})\|\nabla F(\bm{\theta}_t)\|^2 \\  
  &+ \frac{\nu \eta^2\sigma^2}{2BN} + \cfrac{d\nu \eta^2C^2}{2*2^{2b}N}.
	\end{split}
\end{equation}

Subtracting $F(\bm{\theta}_t)$ from both sides, and for $\eta \le \frac{1}{\nu}$
\begin{align*}
	\mathbb{E}[F(\bm{\theta}_{t+1})]- F(\bm{\theta}_t)&\le -\frac{\eta}{2}\|\nabla F(\theta_t)\|^2 + \frac{\eta\sigma^2}{2BN} + \cfrac{d\eta C^2}{2*2^{2b}N}.
\end{align*}
Applying it recursively, this yields:
\begin{align*}
	\mathbb{E}[F(\bm{\theta}_T)]-F(\bm{\theta}_0) &\le -\frac{\eta}{2} \sum_{t=0}^{T-1} \|\nabla F(\bm{\theta}_t)\|^2+ \frac{T\eta\sigma^2}{2BN} + \cfrac{Td\eta C^2}{2*2^{2b}N}.
\end{align*}
Considering that $F(\bm{\theta}_T) \ge F(\bm{\theta}^*)$, so:
\begin{align*}
	&\frac{1}{T}\sum_{t=0}^{T-1}\mathbb{E}[\|\nabla F(\bm{\theta}_t)\|^2] \le \frac{2[F(\bm{\theta}_0)-F(\bm{\theta}^*)]}{T\eta}+ \frac{\sigma^2}{NB} + \cfrac{dC^2}{N2^{2b}}.
\end{align*}
\end{proof}

\subsection{Algorithmic Complexity}
Client-side computation involves three primary tasks:
\begin{itemize}
    \item \textbf{Quantization:} The complexity of quantizing a single element using the half-dithered quantizer is constant, denoted as $O(1)$. As this process is independently applied to each element of the vector, the total computational complexity becomes $O(m)$, where $m$ is the dimension of the gradient.
    
    \item \textbf{Encryption:} The computational complexity of the matrix-vector multiplication and addition is expressed as $O(mn + m)$, which simplifies to $O(mn)$.
    
    \item \textbf{Decryption:} The complexity of decryption is also $O(mn)$.
\end{itemize}

In summary, the overall computational complexity is $O(mn)$. It's worth noting that if the underlying problem is based on ring-LWE \cite{peikert2014lattice} (though the security of ring-LWE is less established), the overall computational complexity could be further reduced to $O(n\log n)$.
 

%% file: 4.Experiments.tex
\section{Experiments}

In this section, we conduct experiments on MNIST to empirically validate our proposed FLAG method. The MNIST consists of 70000 $1 \times 28 \times 28$ grayscale images in 10 classes.

\textbf{Experimental Setting.} The security of our LWE instance is parameterized by the tuple $(n, q, b)$ where $n$ is the dimension of the secret $\bm{s}$, $q$ is the size of modulus, and $b$ controls the width of the uniform distribution $U(-\frac{q}{2^{b+1}},\frac{q}{2^{b+1}})^m$. We set $n=256, m = 768$, $q=65536$, $b=6,8,10$ respectively. Additionally, we list the accuracy achieved by Vanilla Fedearted Learning (Vanilla FL) without gradient compression and encryption as a benchmark.
We conduct experiments for $N = 8$ clients and use LeNet-5~\cite{lecun2015lenet} for all clients. We select the momentum SGD as an optimizer, where the learning rate is set to 0.01, the momentum is set to 0.9, and weight decay is set to 0.0005. We use the $l_{\infty}$ norm of the aggregation gradient from the previous iteration as the clipping threshold $C$ for the current iteration. Considering that $d>>m$ in deep learning, reshaping tensors if necessary, a bucket will be defined as a set of $m$ consecutive vector values. (E.g. the $j$-th bucket is the sub-vector $\bm g^{(i)} [(j-1)m+1 : jm]^\top$). We will encrypt each bucket independently, using FLAG.

Figure~\ref{fig:testing_performance} illustrates the test accuracy of Vanilla FL and FLAG on MNIST. Vanilla FL achieves a test accuracy of 0.9691 with 32-bit full precision gradients. When $b=6,8,10$ bits, FLAG achieves test accuracies of 0.8487, 0.9026 and 0.9338, respectively.

\begin{figure}[ht] 
    \label{fig:Model_Performance}
    \centering
    \includegraphics[width=0.35\textwidth]{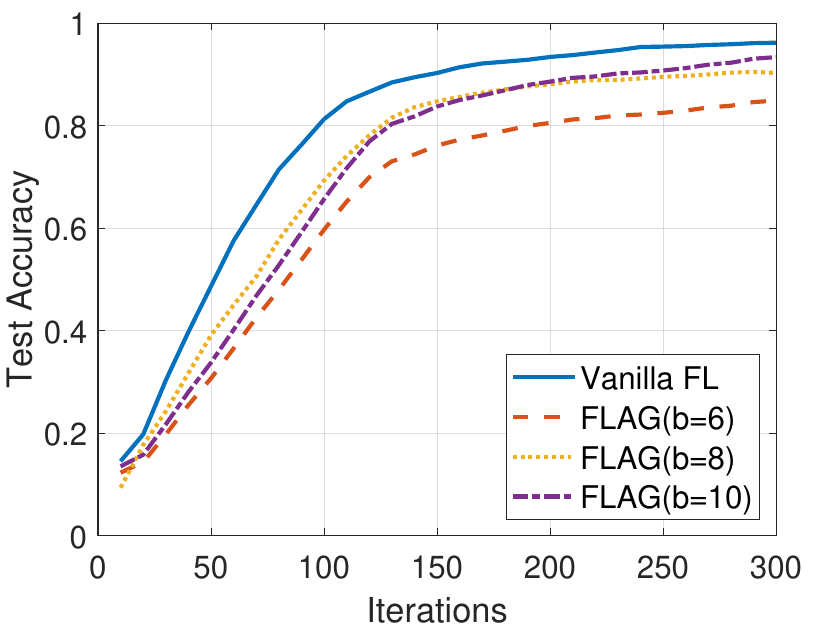}
    \caption{Model Performance of Different Bits.}
    \label{fig:testing_performance}
\end{figure}

Table~\ref{table:Model_Performance} shows the relationship between quantization bits and test accuracy. The proposed FLAG model exhibits a communication-learning tradeoff; that is, the higher the quantization bits, the higher the test accuracy. Additionally, when $b=6,8,10$ bits, the increased communication factors $\tau = 3.5, 2.876, 2.5$ respectively. This means that more quantization bits can help to dilute the additional communication costs generated by encryption and reduce the increased communication factors $\tau$.

\begin{table}[ht] 
\label{table:Model_Performance}
\centering
\caption{Model Performance of Different Bits. }
\begin{tabular}{llll}
\hline
 Bits $b$& 6 &  8 & 10 \\
 \hline
 Increased Factor $\tau$& 3.5 & 2.876 & 2.5 \\
 \hline
 Test Accuracy & 0.8487 & 0.9026 & 0.9338\\
 \hline
\end{tabular}
\end{table}

Table~\ref{table:Model_Performance2} compares the performance of FLOP vs. the performance of an LWE based encryption whose error term is not eliminated (i.e., $\mathrm{Enc}_{\bm s}^{LWE}(\bm g) = \bm {A s} + \gamma \mathrm{Q}_b({\bm g})+\bm{e}$). We adopt the same parameter setting of FLOP as before, except that the number of client is set as $N=100$. The standard deviation of the error $\bm{e}$ in LWE is set as $\sigma=q/(2^b)\times \sqrt{12}$. Notably, when $b=6$, the overflow probability of $\mathrm{Enc}_{\bm s}^{LWE}$ is as large as $0.3558$.
\begin{table}[ht] 
\label{table:Model_Performance2}
\centering
\caption{Probability of overflows. }
\begin{tabular}{llll}
\hline
 Bits $b$& 6 &  8 & 10 \\
 \hline
 FLOP & 0 & 0 & 0 \\
 \hline
$\mathrm{Enc}_{\bm s}^{LWE}$ & 0.3558 & 0.0003 & 0 \\
 \hline
\end{tabular}
\end{table}

%% file: 5.Conclusion.tex
\section{Conclusion}
In conclusion, we have introduced FLAG as  a novel federated learning framework that leverages private-key encryption based on lattices. The error term in LWE is generated from a randomized quantization of the gradients. The CPA security of the scheme is proved based on the hardness of LWE over uniform errors, thus our system leaks no information of participants
to the honest-but-curious parameter server. 
FLAG features a small probability of overflow, and achieves  accuracy to close to unquantized DSGD. The experiments shed light on the adaptability of FLAG under different security parameters and highlight its potential for privacy-preserving federated learning. 